\documentclass[12pt,oneside]{article}
\usepackage[english]{babel}
\usepackage{amsmath,amssymb}
\usepackage{graphicx}
\usepackage{floatrow}
\usepackage{theorem,natbib}
\usepackage{cleveref}
\usepackage[usenames, dvipsnames]{color}
\usepackage[shortcuts]{extdash}

\newtheorem{theorem}{Theorem}[section]

\newtheorem{proposition}[theorem]{Proposition}
\newtheorem{lemma}[theorem]{Lemma}
{\theorembodyfont{\rmfamily}

} {\theorembodyfont{\rmfamily}

}
\newenvironment{proof}[1][]
   {\par\medbreak{\noindent\bfseries Proof#1\quad}}
   {\hbox{}\hfill $\square$\bigbreak}

\newcommand{\familyof}[1]{{\phi #1}}
\newcommand{\averagebundle}{\overline{e}}
\newcommand{\equalsplit}{\mathbf{\bar{e}}}
\def\endowment{e}

\newcommand{\allocation}{\mathbf{x}}

\newcommand{\pricevector}{\mathbf{p}}
\newcommand{\bundles}{\mathbb{R}_+^G}

\ifdefined\ARXIV
\newcommand{\er}[1]{#1}
\newcommand{\erel}[1]{}
\newcommand{\sophie}[1]{}
\else
\newcommand{\er}[1]{\textcolor{ForestGreen}{#1}}
\newcommand{\erel}[1]{\textcolor{ForestGreen}{(Erel: #1)}}
\newcommand{\sophie}[1]{\textcolor{blue}{(Sophie: #1)}}
\fi

\usepackage{caption}
\title{Fairness for Multi-Self Agents}
\author{Sophie Bade and Erel Segal-Halevi}

\begin{document}
\maketitle

\begin{abstract}
We investigate whether fairness is compatible with efficiency in economies with multi-self agents, who may not be able to integrate their multiple objectives into a single complete and transitive ranking. 
%We  revisit the classic question of fair and efficient allocation,
%in the context of behavioral welfare economics, where agents may hold   conflicting objectives.
We adapt envy-freeness, egalitarian-equivalence and the fair-share guarantee in two different ways. An allocation is \emph{unambiguously-fair} if it satisfies the chosen criterion of fairness according to every objective of any agent; it is \emph{aggregate-fair} if it satisfies the criterion for some aggregation of each agent's objectives.

While efficiency is always compatible with the unambiguous fair-share guarantee, it is incompatible with unambiguous envy-freeness in economics with at least three agents. Two agents are enough for efficiency and unambiguous egalitarian-equivalence to clash.
Efficiency and the unambiguous fair-share guarantee can be attained together with aggregate envy-freeness, or aggregate egalitarian-equivalence.
\end{abstract}

\section{Introduction}
We consider multi-self agents who pursue a variety of potentially conflicting objectives.  An agent may, for example, base her decisions on her career prospects, her family life and her immediate enjoyment.
A fully rational agent is able to aggregate all her objectives into a single complete and transitive ranking; a boundedly rational agent may fail to do so.
In an economy with such boundedly rational agents, we consider three criteria of fairness.
An allocation has the \emph{fair-share guarantee (FS)}%
\footnote{Also known as \emph{proportionality}.}
 if each agent prefers their bundle to the average bundle; it is \emph{envy free (EF)} if each agent
 prefers their bundle to any other agent's;
it is
\emph{egalitarian-equivalent (EE)} if each agent is indifferent between their bundle and some fixed reference bundle.

For each fairness criterion,
we call an allocation
\emph{unambiguously fair} if it satisfies the fairness criterion according to each objective of each agent.%
%\footnote{Using the terminology from  \citet{BernheimRangel2009}, an allocation is unambiguously fair if it is fair according to each agent's \emph{unambiguous choice relation}.
%}
For an illustration,  consider agents that are driven by two selves; one with a cool-headed long-run view and another
that greedily searches immediate gratification.\footnote{
\citet{kreps1979representation} and
\citet{gul2001temptation, gul2004self} introduced such agents together with proposals for their complete and transitive rankings over choice sets.}  An allocation among such ``tempted'' agents is unambiguously envy-free if neither self of any agent envies any other agent. In this case
no agent envies any other agent according to any
%--- transitive or intransitive ---
aggregation of her two objectives. The agents may then  avoid the mental load of weighing their long term goals against their immediate greed.

With rational agents, each of the three concepts of fairness is --- under generic conditions --- compatible with efficiency. We then ask whether unambiguous fairness is compatible with efficiency when agents are boundedly rational. The assumption of bounded rationality has two countervailing effects: On the one hand, the set of Pareto optima can be  large. Unambiguous fairness is, on the other hand, hard to satisfy.  A-priori, it is not clear which of these two effects dominates. We find in
Theorem \ref{positive-pofs}  that efficiency is always compatible with the unambiguous fair-share guarantee.
Theorem \ref{positive-pone-2}
shows that any economy with just two agents, all whose selves have convex preferences, has an unambiguously envy-free Pareto optimum.
Unambiguous no-envy, however, conflicts with efficiency in economies with three or more agents (Proposition \ref{negative-pone}).
Two agents suffice for unambiguous egalitarian-equivalence to conflict with efficiency
(Proposition \ref{negative-poee}).

Faced with the non-existence of unambiguously envy-free or egalitarian-equivalent Pareto optima, we investigate  these fairness criteria  according
to some aggregation of all agents' objectives into rational preferences.  Theorems \ref{positive-collective-ef} and \ref{positive-collective-ee} 
show that some Pareto optima with the unambiguous fair-share guarantee respectively satisfy aggregate no-envy and aggregate egalitarian-equivalence, when all selves have strictly convex preferences. 
%state broad conditions under which some Pareto optima with the unambiguous fair-share guarantee respectively satisfy aggregate no-envy and aggregate egalitarian-equivalence. 
The aggregators used in the two results
do not vary with the economies under consideration. To understand these aggregators, interpret the agents and their selves as
 as families and their members.
If the Rawlsian  criterion of justice is used to
resolve intra-family conflict,
the leximin aggregator used
in  Theorem \ref{positive-collective-ef} arises.
If the family instead aggregates its members' objectives using Nash bargaining (maximizing the product of members' utilities), we obtain the aggregator used in Theorem  \ref{positive-collective-ee}.

%  the agent prefers an option only if it yields higher expected utility according to each one of his priors. There are many ways in which the agent can aggregate his objectives into a transitive and complete preference:
% Our agent could for example   evaluate options according to the minimal expected utility over all his priors. In this case his choices are rationalized by Gilboa and Schmeilder's Minimax expected utilities. Alternatively our agent could be an expected utility maximizer using weighed average of all his priors.

%uncertain agents whose different objectives are represented by expected utilities using different priors\footnote{Bewley's model of Knightian uncertainty represents uncertain agents using multiple priors}. In this context the set of unambiguously fair allocations shrinks as the sets of the agents' priors increase.  At the same time though the set of Pareto optima increases, as fewer allocations are unambiguously Pareto improving. The combination of the

When agents are rational, market equilibria from equal endowments have the fair-share guarantee and are envy-free  \citep{foley1967resource}. 
%With rational agents we can therefore use results on the existence of market equilibrium to derive the existence of envy free Pareto optima with the fair-share guarantee. 
So when agents are rational, envy-free Pareto optima with the fair share guarantee exist wherever there are market equilibria.  
With boundedly rational agents this strong nexus between no envy and market equilibrium
 disappears:
 Proposition \ref{negative-ceei} shows that, even if some unambiguously envy free Pareto optima arise as market equilibria, it may be necessary to
 give less rational agents larger endowments to obtain such allocations in market equilibrium.

% \erel{Consider moving the example to the relevant subsection, to shorten the introduction.}
 %For an example, consider an economy where bundles consist of consumption today and tomorrow. Say there are three agents. The first two have standard (rational) time preferences. One of these first two agents cares much more about today's consumption than the other.
% The third agent has two frames: one of these frames assigns more weight to consumption today, the other assigns more weight to consumption tomorrow.
%Agent 3 then
% needs to spend at least as much on first (second) period consumption as do the agents who only spend on only one of the two.
%So markets with equal endowments do not guarantee fair outcomes in economies with boundedly rational agents.

\section{Preliminaries}
\subsection{Goods and Agents}
There are
 $G$ different homogeneous divisible goods and a finite set of agents $I$.
 The set of consumption bundles is $\bundles$; the
 total endowment is $\endowment\in\bundles$ with $\endowment\gg0$.\footnote{For any two vectors $x,x'\in\mathbb{R}^m$ for some integer $m$, say $x\geq x'$ if $x_i\geq x'_i$ for all $i$, $x>x'$ if $x\geq x'$ but not $x=x'$ and $x\gg x'$ if $x_i>x'_i$ for all $i$.}
 An \emph{allocation} is a vector $\allocation = \{x_i\}_{i\in I}$ of consumption bundles $x_i\in \bundles$ whose sum does not exceed the total endowment: $\sum_{i\in I} x_i\leq \endowment$.
$X$ is the set  of all allocations.
The average bundle $\averagebundle{}:=\frac{1}{\mid I\mid}\endowment$ defines each agent's \emph{fair-share} of the total endowment, and each agent gets the fair share in the 
the equal split
$\equalsplit = (\averagebundle,\ldots,\averagebundle)$.

 \subsection{Selves and Aggregators}
When evaluating allocations, agents only consider their own bundles.
Each agent $i$ has a set of \emph{selves} $S_i$.\footnote{Alternatively, one could think of the agent's selves as his \emph{frames} in the sense of \cite{SalantRubinstein2008}.} The preferences of each such self $s\in S_i$ are represented by a continuous utility function $u_s: \bundles\to\mathbb{R}$, normalized  such that $u_s(\averagebundle)=0$. We assume throughout that $u_s$ is strictly increasing in all components, so all selves have strictly monotone preferences.
The preferences represented by
the utility $u_s$ are  \emph{convex} if $u_s(x) > u_s(x')$ implies   $u_s((1 - \alpha)x + \alpha x') > u_s(x')$ for all $x,x'\in \bundles$ and $\alpha\in(0,1)$. They are \emph{strictly convex} if  $u_s((1 - \alpha)x + \alpha x') > u_s(x')$ also holds for any two different bundles $x$ and $x'$ with $u_s(x)=u_s(x')$.

We say that an agent $i$ is \emph{more rational than} an agent $j$ if $S_i\subsetneq S_j$.%
\footnote{The relation ``more rational than-'' is incomplete. For example, a fully rational agent  $i$ with a single self $S_i=\{u_s\}$ is not more-rational than an agent $j$ with two selves if $u_s\notin S_j$.} If an economy is derived from another economy by (weakly) increasing each agent's set of selves, then the original economy is \emph{more rational than} than the derived one. 

Agent $i$ \emph{unambiguously prefers} an option if all his selves prefer this option. Formally,
agent $i$'s \emph{unambiguous preference} $\succsim^U_i$ is a --- typically incomplete --- transitive
 relation, where $x\succsim^U_i x'$ holds if and only if
    $u_s(x)\geq u_s( x')$ for all $s\in S_i$ and where $x\succ^U_i x'$ holds if in addition $u_{s'}(x)> u_{s'}( x')$ holds for some $s'\in S_i$.\footnote{
\citet{mandler2014indecisiveness,mandler2020distributive} calls $\succsim^U_i$ the \emph{behavioral preference}.
\citet{BernheimRangel2009} define a similar relation based on revealed choices, denote it by $R'$ and call it \emph{weak unambigous choice preference}.
}
A transitive and complete preference $\succsim^{agg}_i$ is an \emph{aggregator} of agent $i$'s selves if for all bundles $x,x'$, $x\succsim^U_i x'$ implies $x\succsim^{agg}_i x'$
and $x\succ^U_i x'$ implies $x\succ^{agg}_i x'$.
\citet{szpilrajn1930extension}'s extension theorem guarantees the existence of such aggregators. An agent unambiguously prefers some bundle $x$ to a different bundle $x'$ ($x\succsim_i^U x'$)  if and only if he prefers $x$ to $x'$ ($x\succsim_i^{agg} x'$) according to every aggregator $\succsim_i^{agg}$.%
\footnote{The proof is simple and we omit it. It is available upon request.}
We also call a vector of aggregators $\succsim^{agg}\colon=(\succsim^{agg}_i)_{i\in I}$ an aggregator.

The function  $g((u_s(\cdot))_{s\in S_i}):X\to \mathbb{R}$ represents an aggregator for agent $i$ if $g:\mathbb{R}^{\mid S_i\mid }\to \mathbb{R}$  is strictly increasing in all its components. If $g$ is the sum of its components, then $g$ represents the \emph{utilitarian aggregator}. The \emph{leximin aggregator} is not representable by any function.
% \citep{debreu1954representation}.
To define it, fix any two bundles $x$ and $x'$, and consider the two
minimal utilities of any selves of the agent: if $\min_{s\in S_i}u_s(x)>\min_{s\in S_{i}}u_s(x') $ then $x\succ^{lex}_i x'$. If these two minimal utilities are identical, then the leximin aggregator uses the two second-lowest  utilities to rank $x$ and $x'$. Proceeding inductively $\succsim^{lex}_i$ ranks $x$ and $x'$ as indifferent if and only if both are mapped to the same vector of utilities in increasing order of utility
\citep{Dubins1961How,Moulin2004Fair}.%
\footnote{
The definitions of the utilitarian and leximin aggregators depend on the representations $u_s$ of the selves' preferences.
The income-equivalence principle \citep{decancq2015happiness,FleurbaySchokkaert2013} provides a
specific normalization  with  $u_s(\averagebundle)=0$.
}
To illustrate such aggregators at the hand of a specific decision theoretic model, consider the \cite{Bewley} model of Knightian uncertainty, where the different selves of an agent use different priors to evaluate uncertain events. \cite{Bewley} axiomatized the unambiguous preference $\succsim^U_i$ where agent $i$ prefers an option to a different one if the former yields a higher expected utility than the latter according to every prior of the agent.
If such an agent uses the utilitarian aggregator, he is an expected utility maximizer. If he instead evaluates every choice according to the minimal expected utility over the set of all priors, he is represented by a maximin expected utility, following \citet{gilboa1989maxmin}.%
\footnote{The proof of Theorem \ref{positive-collective-ef} addresses the difference between the leximin aggregator and the function that uses the minimal utility of any self to evaluate options.}

%Viewing the decision theoretic setup through the lens of Knightian uncertainty, so that different frames represent different priors, the utilitarian aggregator is none other than an expected utility. Conversely the leximin aggregator - nearly - coincides with Gilboa and Schmeidler's \cite{} maximin expected utilities where agents only evaluate outcomes according to the minimal expected utilties according to any of their priors.
%Under the interpretation that agents as families, the aggregator $\succsim^{Nash}_i$ represents the choices of a family which uses Nash bargaining with $\averagebundle$ as the outside option to mediate internal conflict.
%The leximin aggregator instead represents Rawlsian families that place the focus on their unhappiest members.

\subsection{Notions of Fairness}
An allocation satisfies the \emph{fair-share guarantee} if each agent weakly prefers his bundle to the fair-share; it
is \emph{envy-free} if no agent strictly prefers the bundle of another agent; it is \emph{egalitarian equivalent} if each agent is indifferent between his bundle and some fixed reference bundle $r$. For a given aggregator $\succsim^{agg}$ and a given fairness notion, we call an allocation  $\succsim^{agg}$-\emph{aggregate}-fair if each agent $i$ considers it fair according to their aggregator $\succsim^{agg}_i$.
Aggregate fairness coincides with standard fairness for rational agents whose preferences are represented by the aggregators. So existing tools suffice to study aggregate fairness.
But
there is a drawback:  we have to commit to particular aggregators.
To avoid such a dependence, we study unambiguously-fair allocations, which are fair according to every aggregator.

%(equivalently: it holds according to the unambiguous relation $\succsim^U_i$ of every $i\in I$).
Specifically, an allocation $\allocation$ is
\emph{unambiguously envy-free (EF)} if there are  no two agents $i,i'$ and self $s\in S_i$ such that
$u_s(x_{i'})> u_s(x_i)$.
 It
is  \emph{unambiguously egalitarian equivalent (EE)}, if there exists a reference bundle $r$  such that for each agent $i$
  $u_s(r)=u_s(x_i)$ holds for all $s\in S_i$. The allocation $\allocation$ satisfies the   \emph{unambiguous fair-share guarantee (FS)} if for each agent $i$ $u_s(x_i)\geq u_s(\overline{e})$ holds for all $s\in S_i$, and $E$ is the set of all unambiguous-FS allocations.\footnote{
For each of the fairness notions, an allocation is unambiguously-fair if and only if it is aggregate-fair according to every vector of aggregators. The proof is available upon request.
}

Unambiguous fairness relies on the transitive but incomplete preferences $\succsim^U_i$ of agents. For a different approach to fairness for boundedly rational agents, we instead consider complete but intransitive aggergators in Section \ref{sec:intrans}.

\subsection{Pareto Optimality}
An allocation $\allocation$  \emph{Pareto dominates}  a different allocation $\allocation'$ if $u_s(x_i)\geq u_s( x'_i)$ holds for every self $s\in S_i$ of every agent $i$, while $u_{s'}(x_j)>u_{s'}( x_j)$ holds for  at least one agent $j$ and self $s'\in S_j$.
An allocation is \emph{Pareto-optimal} it is not Pareto-dominated by any other allocation.

To formalize the two countervailing effects of bounded rationality on the set of fair and efficient allocations, consider a more-rational and a less-rational economy.
Say that some allocation $\allocation'$ Pareto dominates a different allocation $\allocation$ in the less-rational economy, and say that no selves are indifferent between the two allocations.
Then, clearly, $\allocation'$ also dominates $\allocation$ according to the smaller sets of selves in the more-rational economy.%
\footnote{To see that the clause of indifference matters, consider a fully rational economy in which all selves are indifferent between all bundles. Now add one self to one agent's set of selves and assume that this self strictly ranks all allocations. While no allocation dominates any other in the rational economy, this does not hold in the boundedly rational economy.}
So, abstracting away from the issue of indifferences, the set of Pareto-optima increases when the economy becomes less rational.
On the other hand, the set of unambiguously fair allocations in the more-rational economy is larger, since the chosen fairness notion then has to hold for fewer selves.
Our first result shows that --- no matter the level of irrationality --- the set of Pareto optima contains some unambiguously-fair allocations.
\begin{theorem}
\label{positive-pofs}
Any economy has unambiguous-FS Pareto optima.
\end{theorem}
\begin{proof}
Since  $u_s$ is for each $s\in \cup_{i\in N}S_i$ continuous, the set $E$ of unambiguous-FS allocations is compact. Define a function $F:X\to \mathbb{R}$ as the sum of all utilities of all agents' selves: $F(\allocation)=\sum_{i\in N}\sum_{s\in S_i}u_s(x_i)$ for all $x\in X$.
By the compactness of $E$ and the continuity of $F$, some allocation $\allocation^*$ maximizes $F(x)$ over all $x\in E$, so  that $x^*$ is Pareto optimal in $E$. Since no allocation outside $E$ Pareto dominates any allocation in $E$, $\allocation^*$ is Pareto optimal.
Since $x^*\in E$, $x^*$ is unambiguous-FS. 
\end{proof}
The above proof  continues to hold if we use a different method to find a Pareto optimum in $E$ and if 
 we drop the assumption of monotonic preferences.

\subsection{Market Equilibria}
\label{sub:equilibria}
We borrow our notions of market equilibrium from behavioral welfare economics  \citep{FonOtani1979,mandler2014indecisiveness}.
A triplet $(\pricevector,\allocation^*,\allocation^0)$ with $\allocation^*, \allocation\in X$ and $\pricevector\in \mathbb{R}^G$ is  a \emph{market equilibrium from (endowments) $\allocation^0\in X$}, if the following two conditions hold for each agent $i\in I$:
\begin{enumerate}
\item $\pricevector x^*_i \leq \pricevector x^0_i$, that is, each agent $i$ can afford bundle $x^*_i$; and---
\item
\label{cond:affordable}
For any bundle $x'\in\bundles$,
%\er{
%$x'\succ^U_i x_i^*$ implies
%$\pricevector  x' > \pricevector x_i^0$.
%}
%\er{In other words:}
if $\pricevector x' \leq \pricevector x^0_i$ then
either $u_s(x')=u_s(x^*_i)$ for all selves  $s\in S_i$,
or $u_s(x')<u_s(x^*_i)$ for at least one self $s\in S_i$.
\end{enumerate}
The second condition is weak:  agent $i$'s selves may not unanimously prefer a different affordable bundle to the choice $x^*_i$.
Equivalently, any agent may buy any bundle that is optimal according to \emph{some} aggregation of her selves.
So the definition remains agnostic as to the aggregator an agent uses when she shops. Our result on the impossibility of attaining fairness in market equilibrium (Proposition \ref{negative-ceei}) is strengthened by this agnosticism.%
\footnote{If we would instead require that any self of  agent $i$ prefers $x^*_i$ to all affordable bundles, nonexistence of market equilibria would be immediate, as an agent's selves need not agree on optimal choices.
 Using aggregate preferences in the second condition  would tie us down to a particular rational theory on how agents weigh their different selves.}

%\begin{example}
%[Comparing bundles]
%\label{example: preferences}
%Consider a family $f$ consisting of a husband and a wife. Say two different bundles $x$ and $x'$ respectively yield utility vectors $(1,2)$ and $(4,1)$ to the husband and the wife.
%According to the aggregation represented by $U_f(x) := \sqrt{u_h(x)}+u_w(x)$ for all $x\in\mathbb{R}^G_+$, the family
% is indifferent between the two bundles. Since different family members prefer different bundles, we have $x\bowtie^{col}_f x'$, so the family's collective preference over the two bundles is incomplete. The family cannot unambiguously rank the two bundles. Conversely if the two bundles $x$ and $x'$ respectively yield utility vectors $(9,1)$ and $(3,1)$ then the family prefers bundle $x$ according to all notions of preference defined above.

%\end{example}

\subsection{Examples with two goods}
\label{sub:two-goods}
 In  economies with two goods,
these are called $y$ and $z$, and we normalise the total endowment of  each good to $|I|$, so that $\averagebundle=(1,1)$.
For a differentiable utility $u_s:\mathbb{R}^2_+\to \mathbb{R}$, the \emph{marginal rate of substitution} between $y$ and $z$ is:
\begin{align*}
MRS_s(y,z) := \frac{d(u_s(y,z))}{dy}\bigg/\frac{d(u_s(y,z))}{dz}.
\end{align*}

Two different utilities $u_s$ and $u_{s'}$ have the \emph{single crossing property}
if any two indifference-curves defined by $u_s(y,z)=\alpha$ and $u_{s'}(y,z)=\beta$ for some $\alpha,\beta\in \mathbb{R}$ share at most one point. So
 $u_s$ and $u_{s'}$ have the  single-crossing property if for all bundles $(y,z)$ and $(y',z')$ and either $u_s=u$, $u_{s'}=u'$ or $u_s=u'$, $u_{s'}=u$:
\begin{align*}\mbox{If }y'> y\mbox{ and }z'<z\mbox{ then }
u(y',z')\geq u(y,z)\Rightarrow u'(y',z')> u'(y,z)\\
\mbox{If }y'< y\mbox{ and }z'>z\mbox{ then }
u'(y',z')\geq u'(y,z)\Rightarrow u(y',z')> u(y,z).
\end{align*}
Two differentiable utilities have the single-crossing property if the marginal rate of substitution of one is higher than the other's at each bundle $(y,z)$.
For example, any two Cobb-Douglas utilities $u_s(y,z)=y^{\alpha}z^{1-\alpha}$ with different $\alpha$
 have the single-crossing property.

Our proofs and examples use the characterization of all interior Pareto-optimal allocations as the set of  allocations at which the intersection of all agents' MRS-ranges, the intervals between the smallest and the largest MRS of their selves,   is non-empty.

\begin{theorem}
\label{po-characterization}
Fix a two-good economy with convex preferences representable by differentiable utilities
and an allocation $\allocation$  in the interior of $X$. Then $\allocation$ is Pareto optimal if and only if:
\begin{align*}
\bigcap_{i\in I}[\min_{s\in S_i}MRS_s(y_i,z_i),\max_{s\in S_i}MRS_s(y_i,z_i)]\neq\emptyset.
\end{align*}
\end{theorem}
Theorem \ref{po-characterization}
has been shown to hold in much more general environments by   \citet{FonOtani1979,mandler2014indecisiveness} and we omit its proof  here. Simple modifications of the arguments in
\citet{mas1974equilibrium} also prove Theorem \ref{po-characterization}.

\section{Envy-freeness}
\label{sec:pone}
Theorem \ref{positive-pone-2}  shows that
two-agent economies with convex preferences always have
unambiguous-EF
Pareto optima.  Economies with more  agents, in contrast, need not have such allocations (Proposition \ref{negative-pone}).
 Even if unambiguous-EF  Pareto optima  exist, they may not arise as market equilibria from equal endowments. The market approach, however, turns out to be useful to find unambiguous-FS Pareto optima that are aggregate-EF.
The impossibility results hold even if all but one agent are fully rational.

\subsection{Unambiguous envy-freeness}
\label{sub:pone-with-two}

\begin{theorem}
\label{positive-pone-2}
Any two\-/agent economy, where all selves have convex preferences, has an unambiguous\-/EF and unambiguous\-/FS
 Pareto optimum.
\end{theorem}

\begin{proof}
By Theorem \ref{positive-pofs}, there exists an unambiguous-FS Pareto optimum $\allocation$. To see that $\allocation$ is unambiguous-EF, suppose by contradiction that some self $s$ of some agent $i$ envies agent $j$, so that $u_s(x_i) < u_s(x_{j}) = u_s(\endowment- x_i)$. By convexity, $u_s(x_i)$ must also be smaller than $u_s(\frac{1}{2}x_i+\frac{1}{2}(\endowment{}- x_i))$.
But the bundle $\frac{1}{2}x_i+\frac{1}{2}(\endowment{}- x_i)$ is exactly the fair-share $\averagebundle$ --- a contradiction to the assumption that $x$ is unambiguous-FS.
\end{proof}

The preceding proof only requires  continuous and convex preferences, the assumption of monotonicity can be dropped.
To show that Theorem \ref{positive-pone-2} does not extend to economies with more agents, we
construct a two-goods economy with two fully rational and one boundedly rational agent. All selves' utilities have the single-crossing property and the preferences of the two fully rational agents 1 and 2 are intermediate between the two selves of the boundedly rational agent 3. To avoid envy between agent 3 and the two fully rational agents,
  agent 3 must consume the same bundle as agents 1 \emph{and} as agent 2. A contradiction then arises since agents 1 and 2 have different marginal rates of substitution at any bundle, and therefore must
  consume different bundles in any interior Pareto optimum.
 Figure \ref{fig:pone-none} illustrates.

\begin{proposition}
\label{negative-pone}
Economics with three agents and  two goods need not have any
unambiguous-EF Pareto optima.
\end{proposition}

\begin{figure}
\floatbox[{\capbeside\thisfloatsetup{capbesideposition={left,top},capbesidewidth=.25\textwidth}}]
{figure}
[\FBwidth]
{\caption{\small
	\label{fig:pone-none}
	Proof of Proposition \ref{negative-pone}. The solid lines 1,2  and dashed lines $h$, $w$ respectively represent the indifference curves of the rational agents and the selves of the agent 3.
	}
}
{
	\includegraphics[width=.7\textwidth]{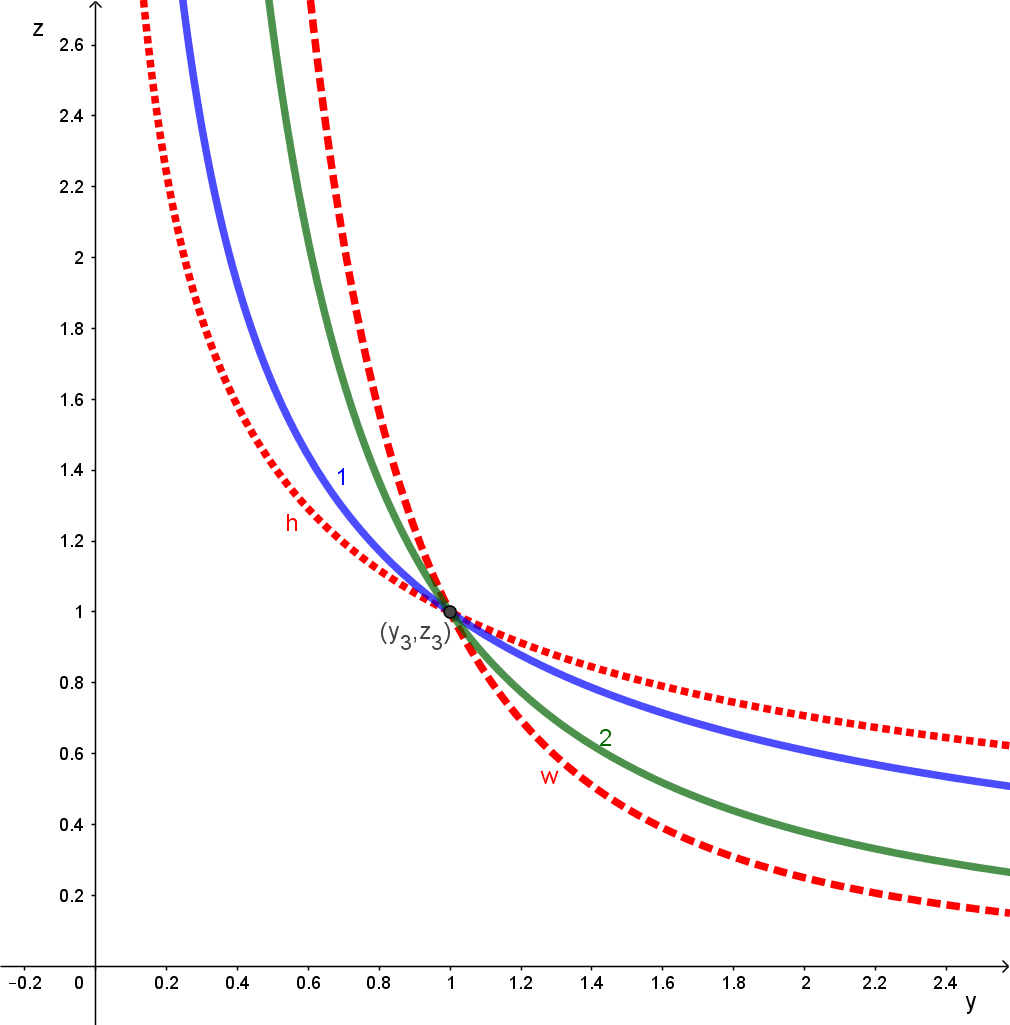}
}
%\begin{center}
%\end{center}
\end{figure}

\begin{proof}
Fix a two-good economy with two fully rational agents 1 and 2 represented by $u_1$ and $u_2$ and a boundedly rational  agent 3 whose two selves are represented by $u_h$ and $u_w$.
All selves'  utilities are
differentiable and convex, and satisfy a single-crossing property so that for any two  bundles $(y,z)$, $(y',z')$:
\begin{align*}
\text{If }&y'> y\text{ and }z'<z\text{ then:}
&u_{h}(y',z')\geq u_h(y,z)\Rightarrow u_{2}(y',z')> u_{2}(y,z),\\&& u_{2}(y',z')\geq u_{2}(y,z)\Rightarrow u_{1}(y',z')> u_{1}(y,z),\\
&& u_{1}(y',z')\geq u_1(y,z)\Rightarrow u_{w}(y',z')> u_{w}(y,z).\\
\text{If }&y'< y\text{ and }z'>z\text{ then: }
&u_{w}(y',z')\geq u_w(y,z)\Rightarrow u_{1}(y',z')> u_{1}(y,z),\\&& u_{1}(y',z')\geq u_{1}(y,z)\Rightarrow u_{2}(y',z')> u_{2}(y,z),\\
&& u_{2}(y',z')\geq u_{2}(y,z)\Rightarrow u_{h}(y',z')> u_{h}(y,z).
\end{align*}
Fix an unambiguous-EF allocation  $\allocation=((y_1,z_1),(y_2,z_2),(y_{3},z_{3}))$.
%For agent $1$ not to envy agent $3$
%$ u_1(y_1,z_1)\geq u_1(y_3,z_3)$ must hold.
Since all preferences are strictly monotonic,  either  $(y_1-y_3)(z_1-z_3)< 0$ or $(y_1,z_1)=(y_3,z_3)$ must hold for agents 1 and 3 not to envy each other:
\begin{itemize}
\item If $y_1> y_3$ and $z_1< z_3$, then $u_1(y_1,z_1)\geq u_1(y_3,z_3)$ and single-crossing imply $u_w(y_1,z_1)> u_w(y_3,z_3)$, so self $w$ of agent 3 envies agent 1.
\item If $y_1< y_3$ and $z_1> z_3$, then $u_1(y_1,z_1)\geq u_1(y_3,z_3)$ and single-crossing imply  $ u_h(y_1,z_1)> u_h(y_3,z_3)$, so self $h$ of agent 3 envies agent 1.
\item So we must have $(y_1,z_1)=(y_3,z_3)$.
\end{itemize}
Mutatis mutandis  $(y_{2},z_{2})=(y_3,z_3)$ also holds and all three agents consume the same bundle. So $\allocation=\equalsplit$. But $\equalsplit$ cannot be Pareto optimal since the two fully rational agents consume different bundles in any interior Pareto-optimum (Theorem \ref{po-characterization}).
\end{proof}

The  economy in the proof of Proposition \ref{negative-pone}  illustrates the two countervailing forces of bounded rationality on the set of unambiguously fair Pareto optima.
By standard results, a fully rational version of this economy, where agent 3 has only one self, has an unambiguous-EF Pareto optimum. With the addition of agent 3's second self, the set of Pareto-optima increases while the set of unambiguous-EF allocations decreases. The second effect dominates: the example in the proof has no unambiguous-EF Pareto optimum.
But with the further addition of
 $\frac{1}{2}u_w+\frac{1}{2}u_h$
to agent 1's and agent 2's sets of selves, the set of Pareto optima increases enough for $\equalsplit$ to be an
unambiguous-EF Pareto optimum in this less-rational economy,
by Theorem \ref{po-characterization}.

The example used to prove Proposition \ref{negative-pone}
can be extended to any number of agents by  adding fully rational agents with preferences that are intermediate between the two selves of agent 3.
To extend it to any number of goods, replace the single good $z$ with  $G-1$ perfect substitutes $z^1,\ldots,z^{G-1}$, and define
 self $s$'s utility for bundle $y,z^1,\ldots,z^{G-1}$
as
$u_s(y, z^1 + \cdots + z^{G-1})$,
where $u_s$ is the utility of $s$ in the two-goods economy.

\subsection{Aggregate envy-freeness}
\label{sub:aggregate-ef}
While Theorem \ref{negative-pone} shows that unambiguous-EF Pareto optima need not exist, standard results show that mild conditions suffice for existence of aggregate-EF Pareto optima. Here we investigate conditions under which some aggregate-EF Pareto optima have the unambiguous fair-share guarantee. 
%The set of Pareto optima that guarantee the fair-share to each self contains---under broad conditions---
When all selves' preferences are strictly convex, the set of Pareto optima that guarantee the fair-share to each self contains 
some allocations that are envy free according to a particular aggregator---the leximin aggregator.

\begin{theorem}
\label{positive-collective-ef}
If all selves' preferences are strictly convex there exist unambiguous\-/FS and $\succsim^{lex}_f$\-/aggregate\-/EF
 Pareto optima.
\end{theorem}

\begin{proof}
 For each agent $i$ the function   $U^{\min}_i(\cdot) := \min_{s\in S_i}u_s(\cdot)$%
\footnote{Since $U^{\min}_i$ is indifferent between any two bundles associated with the same minimal utility over all selves, $U^{\min}_i$ does not represent an aggregator for agent $i$'s selves. We can therefore not
 directly use    $U^{\min}_i$ as the aggregator in Theorem  \ref{positive-collective-ef}.}
 inherits the continuity, monotonicity and strict convexity of all selves' preferences.\footnote{
\label{ftn:uimin}
To see the strict convexity, fix any two bundles $x\neq x'$, an $\alpha \in(0,1)$, and an agent $i\in I$.
Suppose that $U^{\min}_i(x)\geq U^{\min}_i(x')$.
Let $s^*\in S_i$ be the self for which the latter minimum is attained, that is,
$U^{\min}_i(x') = u_{s^*}(x')$.
 So we have for all $s\in S_i$
\begin{align*}
u_s(x')\geq u_{s^*}(x')\text{ and }
u_s(x)\geq U^{\min}_i(x) \geq U^{\min}_i(x') = u_{s^*}(x').
\end{align*}
Since each $u_s$ represents a strictly convex preference
 we have
$u_s((1-\alpha)x + \alpha x') >  u_{s^*}(x')$ for all $s\in S_i$, which implies
 $U^{\min}_i((1-\alpha)x + \alpha x') > u_{s^*}(x') =  U^{\min}_i(x')$.}
Say   $(\pricevector,\allocation^*,\equalsplit)$  is  a market equilibrium from equal endowments in the economy where each agent $i$'s preferences  are represented by $U^{\min}_i$. 

\medskip

\textbf{Claim 1: $\allocation^*$ is unambiguous-FS.} Fix an arbitrary agent $i$. Thanks to the normalization $u_s(\averagebundle)=0$, and since
$\averagebundle$ is in each agent's budget set,  $\min_{s\in S_i}u_s(x^*_i)=U^{\min}_i(x^*_i)\geq U^{\min}_i(\averagebundle)= \min_{s\in S_i}u_s(\averagebundle)=0$.
So  $u_s(x^*_i)\geq 0$  holds for each $s\in S_i$, and $\allocation^*$ satisfies unambiguous-FS.

\textbf{Claim 2:  $\allocation^*$ is $\succsim^{lex}_f$-aggregate-EF.} Fix an arbitrary agent $i$.
Since $\allocation^*$ is an equilibrium allocation, $x^*_i$ is $U^{\min}_i$-maximal in agent $i$'s budget set. Since budget sets are convex and since
 $U^{\min}_i$ represents strictly convex preferences
 $x^*_i$ is the
 unique $U^{\min}_i$-maximum.
By the definition of $\succsim^{lex}_i$, the set of $\succsim^{lex}_i$-maxima is a subset of the set of $U^{\min}_i$-maxima, so $x^*_i$ is also the unique $\succsim^{lex}_i$-maximum in the budget set.  Since all agents have identical budgets, $x^*_i$ is $\succsim^{lex}_i$ preferred to any $x^*_j$, and
the allocation $\allocation^*$ is  $\succsim^{lex}$-aggregate-EF.

\textbf{Claim 3:  $\allocation^*$ is Pareto optimal.}
Suppose that some allocation $\allocation'$ Pareto-improves on $\allocation^*$, so that
$u_s(x'_i)\geq u_s(x^*_i)$ for all $s\in S_i$ and $i\in I$, and $u_{s'}(x'_j)> u_{s'}(x^*_j)$ for some
$s'\in S_j$ and $j\in I$. The definition of $\succsim^{lex}_i$ then yields
$x'_i \succsim^{lex}_i x^*_i$ for all $i\in I$,
and   $x'_{j}  \succ^{lex}_{j} x^*_{j}$ for $j\in I$.
By the arguments in Claim 2,
$x^*_i$ is for each agent $i$ the unique $\succsim^{lex}_i$-maximal choice in their budget set.
So
$\pricevector^*\cdot x'_{j} > \pricevector^*\cdot x^*_{j} $ holds for any agent $j$ with $x'_{j}  \succ^{lex}_j x^*_{j}$ while $x'_{i}=x^*_{i}$ holds for all other agents $i$. The monotonicity of  agents' preferences implies
$\pricevector^* \cdot x^*_i
= \pricevector^*\cdot\averagebundle$ for any agent $i$. Summing over all agents, yields  $\sum_{i\in I} \pricevector^*\cdot x'_i > \pricevector^*\cdot \endowment$ -- a contradiction to the
 feasibility of $\allocation'$ which requires
$\sum_{i\in I} x'_i \leq \endowment$.
\end{proof}

The requirement that allocations satisfy no-envy according to at least one self of each agent would appear to be weaker than $\succsim^{lex}$-aggregate no envy.
To see that this is not the case, we construct a two-goods economy where no unambiguous-FS Pareto optimum is one-self envy-free.
Assume two  rational agents with different Cobb-Douglas utilities and a boundedly rational agent each of  whose two  selves  only cares about one good.  Say $\allocation^*$ was an unambiguous-FS Pareto optimum satisfying no-envy according to at least one self of each agent. By unambiguous-FS, agent 3's selves prefer $(y^*_3,z^*_3)$ to the fair-share $\averagebundle=(1,1)$, so $y^*_3\geq 1$ and $z^*_3\geq 1$.
W.l.o.g., say that no envy holds for the self of agent 3  that only cares about $y$, so that $y^*_3\geq y^*_2,y^*_1$. For agents 1 and 2 not to envy agent 3, we then must have $z^*_1,z^*_2\geq z^*_3$. By the single crossing property  agents 1 and 2 must consume different bundles. The agents, in sum, consume more than three units of $z$ --- a contradiction.

While the two selves of  agent 3 violate our standard monotonicity assumption, nearby economies with strictly monotonic preferences will, by continuity, also lack unambiguous-FS Pareto optima that are one-self envy-free.
To construct such a nearby economy, represent the preferences  agent $3$'s  two selves by
$y_3+\epsilon \sqrt{z_3}$ and $z_3+\epsilon\sqrt{y_3}$ for some small $\epsilon>0$.%
\footnote{A similar example shows that we cannot replace the leximin aggregator in Theorem \ref{positive-collective-ef} with an arbitrary aggregator: there is no unambiguous-FS aggregate-EF Pareto optimum, according to the aggregator where agent 3 ranks any bundle with more $y$ above any bundle with less $y$ and considers $z$ only to compare two bundles with equally much $y$.
Moreover, for every real constant $\kappa$,
suppose that the preferences of agent $i$ are aggregated using the function $\sum_{s\in S_i}u_s^{\kappa} / \kappa$ (this is equivalent to the utilitarian aggregator for $\kappa=1$, approaches the Nash aggregator for $\kappa=0$, and approaches the leximin aggregator for $\kappa\to-\infty$; see \citet{Moulin2004Fair}).
Then, a similar example shows that there may be no unambiguous-FS aggregate-EF Pareto optimum for any $\kappa>-\infty$.
}

\subsection{Market equilibrium and unambiguous envy-freeness}
\label{sec:pone-without-ceei}

The classic proof that envy-free Pareto optima exist \citep{foley1967resource} argues that market equilibria from equal endowments yield such allocations. Since the fair-share is  each agent's endowment, each agent must weakly prefer the bundle she buys in equilibrium  to the fair-share.
Since each agent faces the same budget set, each agent could choose any other agent's bundle. So
no agent  envies any other.
This technique does not fare as well in economies with boundedly rational agents.
Proposition \ref{negative-pone} already shows that unambiguous-EF Pareto optima need not exist. But even
even if such
 allocations exist, they may not arise as market equilibria from equal endowments.

\begin{proposition}
\label{negative-ceei}
(a)
Some economies
have unambiguous\-/EF and unambiguous\-/FS Pareto optima, even though no such allocation arises as a market equilibrium
from equal endowments.

(b) In any economy,
if an unambiguous\-/EF allocation $\allocation$ can be sustained as a market equilibrium, and if agent $i$ is more rational than agent $j$, then agent $j$'s equilibrium budget must be at least as large as agent $i$'s.
\end{proposition}
\begin{proof}
(a)
Fix a three-agent two-good economy. Say the two fully rational agents 1 and 2  are represented by Cobb-Douglas utilities with respectively $\alpha_1=\frac{1}{3}$ and $\alpha_2=\frac{2}{3}$. These utilities also represent the two selves of the boundedly rational agent 3.
We first show that this economy has a  unambiguous-EF Pareto-optimum $\allocation^*=\big((y_1,z_1), (y_{2},z_{2}),(y_3,z_3)\big)$. For agents 1 and 2 not to envy agent 3 and vice versa, $y_1^{\frac{1}{3}}z_1^{\frac{2}{3}}=y_3^{\frac{1}{3}}z_3^{\frac{2}{3}}$
and
$y_{2}^{\frac{2}{3}}z_{2}^{\frac{1}{3}}=y_3^{\frac{2}{3}}z_3^{\frac{1}{3}}$ must hold. If $\allocation^*$  is on the boundary of $X$ then at least one self has zero utility.
  No envy then implies that all  selves have utility zero
  at $\allocation^*$ which can therefore not be Pareto optimal.
So $\allocation^*$
must be in the interior of $X$ and
$\frac{1}{2}\frac{z_1}{y_1}=
MRS_1(y_1,z_1)=MRS_{2}(y_{2},z_{2})
=
2\frac{z_{2}}{y_{2}}$ must hold. Combining the preceding equations with the resource constraints  $y_1+y_{1}+y_3=3$ and
$z_1+z_{2}+z_2=3$ we obtain a system of
5 equations in 6 unknowns. To obtain a concrete solution $\allocation^*$ (illustrated in Figure \ref{fig:pone-none-without-ce}), we additionally impose the symmetry condition $y_3=z_3$. The symmetric solution is:
\begin{align*}
y^*_1 \approx 0.654 && z^*_1\approx 1.308
\\
y^*_2 \approx 1.308 && z^*_2\approx 0.654
\\
y^*_3 \approx 1.038 && z^*_3\approx 1.038
\end{align*}

\begin{figure}
\floatbox[{\capbeside\thisfloatsetup{capbesideposition={left,top},capbesidewidth=.35\textwidth}}]
{figure}
[\FBwidth]
{
\caption[]{
\small
\label{fig:pone-none-without-ce}
Proof of Prop. \ref{negative-ceei}.
The blue discs denote the allocations of the agents in the symmetric unambiguous-EF Pareto-optimum.
The solid green line is an indifference curve of agent $1$ and self $h$ of agent 3; the dashed red line is an indifference curve of agent $2$ and self $w$ of agent 3.
The black dotted line represents the budget constraint with endowment $(0.981,0.981) < \averagebundle$.
}
}
{
\includegraphics[width=.6\textwidth]{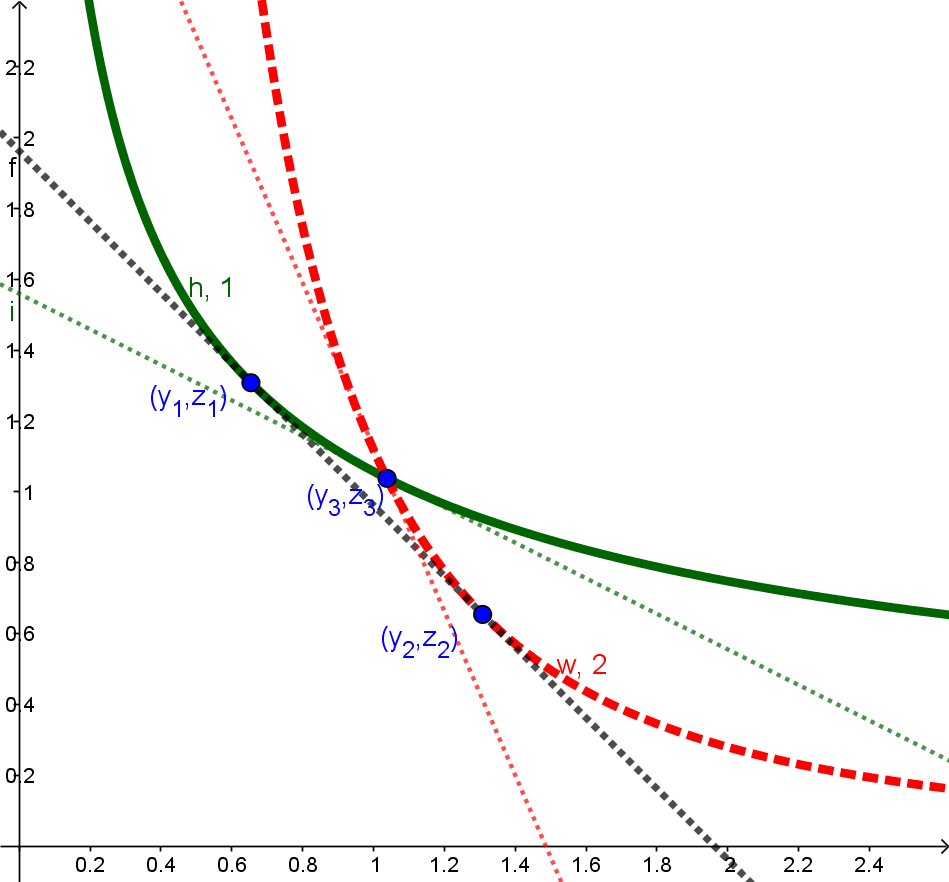}
}
%\begin{center}
%\end{center}
\end{figure}

Since agent 3's bundle in $\allocation^*$ contains  strictly more than the average share from each good, $\allocation^*$ cannot
be obtained as a market equilibrium from equal endowments. So suppose  some other unambiguous-EF Pareto optimum $\big((y_1,z_1)$, $(y_{2},z_{2})$, $(y_3,z_3)\big)$ could be obtained as a market equilibrium from equal endowments.
By Pareto-optimality and the single-crossing property we have $(y_1,z_1)\neq (y_{2},z_{2})$, so that
 $(y_3,z_3)\neq (y_1,z_1)$
or $(y_3,z_3)\neq (y_{2},z_{2})$  (or both).
W.l.o.g. say $(y_3,z_3)\neq (y_1,z_1)$.
Since all selves have strictly convex preferences, the fully rational agent 1 strictly prefers
$(y_1,z_1)$ to all other bundles in the budget set, in particular  $u_1(y_1,z_1)>u_1(y_3,z_3)$.
 But then agent 3 envies agent 1 according to the $u_1$ self --- a contradiction.

 (b) Say $\allocation$ is  unambiguous-EF  and  $(\pricevector,\allocation,\allocation^0)$ is a market equilibrium from some endowment $\allocation^0$.
  Suppose agent $i$'s equilibrium budget was strictly larger than agent $j$'s (so $\pricevector x^0_i>\pricevector x^0_j$).
By definition of market equilibrium, agent $i$'s bundle $x_i$ is optimal in agent $i$'s budget set according to some aggregator $\succsim^{agg}_i$.
Denote by $x^*_i$ the $\succsim^{agg}_i$-optimal choice from agent $j$'s budget set, so that $x^*_i\succsim^{agg}_i x_j$.
Since agent $i$'s budget set strictly contains agent $j$'s, and since all selves preferences (and therefore also $\succsim^{agg}_i$) are strictly monotonic,
$x_i\succ^{agg}_ix^*_i$. 
Since $\allocation$ is unambiguous-EF, $u_s(x_j)\geq u_s(x_i) $ holds for all $s\in S_j$. Since agent $i$ is more rational than agent $j$, we have $S_i\subset S_j$. Since $\succsim^{agg}_i$ is an aggregator of agent $i$'s selves, we in sum get the contradiction $x_j\succsim^{agg}_i x_i\succ^{agg}_ix^*_i\succsim^{agg}_i x_j$.
\end{proof}

Our permissive definition of market equilibrium (Sub. \ref{sub:equilibria}) strengthens the above result. No matter which aggregators the agents happen to use while shopping,   no market equilibrium from equal endowments is unambiguous-EF in the above example.
The result continues to hold for any more restrictive notion of market equilibrium.

To illustrate Proposition \ref{negative-ceei}, consider insurance policies against earthquakes and flooding.  Two fully rational agents with the same income but different priors may buy different insurance policies.
Now consider a boundedly rational third  agent, whose selves use the priors of the two preceding agents, as in the \cite{Bewley} model of Knightian uncertainty.
The third agent then needs a higher income to avoid envy of the other two agents. The more rational agents can better target their income to achieve their goals.
Indeed, the symmetric allocation in part a) of the above proof, can arises as a market equilibrium
when we endow the boundedly rational agent 3 with a higher income than the fully rational agents.

\section{Egalitarian Equivalence}
\label{sec:poee}
\subsection{Unambiguous egalitarian-equivalence}
Proposition \ref{negative-poee}  shows that
unambiguous egalitarian\-/equivalence is even more restrictive than unambiguous envy\-/freeness, in the sense that  even two-agent economies may lack  unambiguous-EE Pareto optima.
The fact that an agent with two selves with the single-crossing property is only indifferent between a reference bundle $r$ and some bundle $x$ if $x=r$ is the driving force behind the upcoming non-existence proof. This proof specifies an economy with two such agents who must  both consume the  reference bundle $r$ in any unambiguous-EE allocation. At the same time, we differentiate the agents enough for them to consume different bundles in any Pareto optimum.

\begin{proposition}
\label{negative-poee}
When there are at least two agents and at least two goods, a unambiguous-EE Pareto optimum might not exist.
\end{proposition}
\begin{proof}
Fix a two-agents two-goods economy. Define   $u_{\gamma}(y,z)\colon=y^{\gamma}+z$. Represent the two selves of agents 1 and 2  by $u_{1/5}$ and $u_{1/4}$, and respectively $u_{1/3}$ and $u_{1/2}$.
Fix an unambiguous-EE Pareto optimum  $\allocation $ with reference bundle $r$. For $\allocation$ to be an unambiguous-EE Pareto optimum, each agent must consume a positive amount of $y$.\footnote{This holds since the marginal utility of $y$ at 0 is for each $u_{\gamma}$ infinite.} Assuming large enough total endowment of $z$, each agent must also consume a positive quantity of $z$ in $\allocation$.
The single-crossing property together with  $u_{1/4}(y_1,z_1)=u_1(r)$ and $u_{1/5}(y_1,z_1)=u_1(r)$ imply $(y_1,z_1)=r$.
By the same token, agent $2$ must also consume $r$, so that $\allocation=(r,r)$.
Since $\max_{s\in S_1} MRS_s(r) > \min_{s\in S_2} MRS_s(r)$ the allocation is by
 Theorem \ref{po-characterization}   not Pareto optimal.
Figure \ref{fig:poee} illustrates.
\end{proof}

\begin{figure}
\floatbox[{\capbeside\thisfloatsetup{capbesideposition={left,top},capbesidewidth=.35\textwidth}}]
{figure}
[\FBwidth]
{
\caption[]{
\label{fig:poee}
\small
Proof of Proposition \ref{negative-poee}.
Agent 1 and 2's two selves are respectively illustrated by the solid and dotted indifference curves.
To be indifferent between their bundle and some $r$, they must consume exactly $r$. But this allocation is Pareto dominated by giving the solid agent $r+(\epsilon,-\epsilon)$ and giving the dotted agent $r+(-\epsilon,\epsilon)$, for some small $\epsilon$.
}
}
{
\includegraphics[width=.6\textwidth]{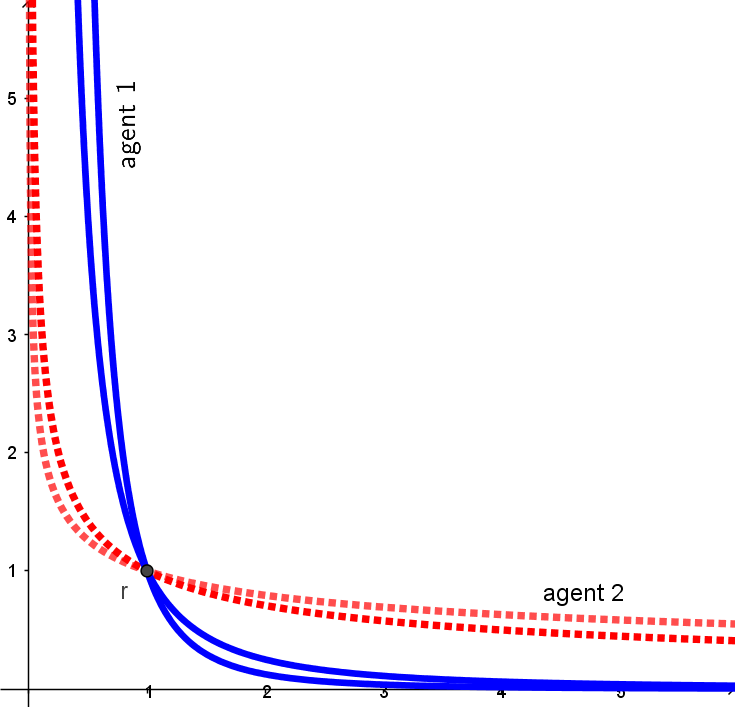}
}
\end{figure}

%
%\begin{figure}
%\begin{center}
%\includegraphics[width=.9\textwidth]{poee}
%\end{center}
%\caption{
%\label{fig:poee}
%Proof of Proposition \ref{negative-poee}.
%Agent 1 and 2's two selves are respectively illustrated by the solid and dotted indifference curves.
%To be indifferent between their bundle and some $r$, they must consume exactly $r$. But this allocation is Pareto dominated by giving the solid agent $r+(\epsilon,-\epsilon)$ and giving the dotted agent $r+(-\epsilon,\epsilon)$, for some small $\epsilon$.
%}
%\end{figure}

The analysis of the above example remains unchanged if we add further agents to the above economy. As long as such agents have multiple selves with the single-crossing property, they must consume the reference bundle $r$ to be unambiguously indifferent between their consumption and $r$. With enough differentiation between the agents they cannot all consume the same bundle in a Pareto optimum.
For an example with more than two goods,
replace the single good $z$ with some $G-1$ perfect substitutes $z^1,\ldots,z^{G-1}$,
as in comment after Proposition \ref{negative-pone}.
As is the case for unambiguous\-/EF Pareto-optima,
the existence of unambiguous\-/EE Pareto-optima is not ``monotone'' in the variability of selves: if we increase each agent's set of selves in the above example to include each other's selves, then the equal split is Pareto\-/optimal and unambiguous\-/EE {by Theorem \ref{po-characterization}.

\subsection{Aggregate  egalitarian-equivalence}
\label{sub:collective-ee}
In parallel to Subsection \ref{sub:aggregate-ef} on envy freeness,  there exist aggregators for which  efficiency, the unanimous fair-share guarantee and aggregate egalitarian equivalence are compatible.
Theorem \ref{positive-collective-ee} shows that such compatibility is achieved for the \emph{Nash aggregator}.
To define this aggregator,  interpret agents and their selves as families and their members.
If these families use Nash-bargaining with the fair-share as the outside option, they choose the  bundle $x^*_i$ that maximizes the product of their members' utilities.

The monotonicity and continuity of all selves' preferences allows us to
 normalize any  $u_s$ such that
 $u_s(x_i)\colon=t-1$ if $u_s(x_i)=u_s(t\cdot \averagebundle)$, where $u_s(\averagebundle)=0$ continues to hold.
  Define
$U^{Nash}_i(x)\colon=\sqrt[\mid S_i\mid]{\Pi_{s\in S_i}u_s(x_i)}$ to represent  agent $i$'s Nash-bargaining aggregator  over bundles in the interior of $E$.
For any $t\geq 1$ we have $U^{Nash}_i(t\cdot \averagebundle)= \sqrt[\mid S_i\mid]{\Pi_{s\in S_i}u_s(t\cdot \averagebundle)}= \sqrt[\mid S_i\mid]{(t-1)^{\mid S_i \mid }}=t-1$. Since $U^{Nash}_i$ only represents an aggregator for bundles that all selves strictly prefer to $\averagebundle$, we use the leximin aggregator for all remaining comparisons to define an aggregator $\succsim_{i}^{Nash}$. Formally, the value of $x\succsim^{Nash}_i x'$ equals ---
\begin{align*}
&
U_i^{Nash}(x)\geq U_i^{Nash}(x') && \text{if~}  u_s(x),u_s(x')>0 \text{~for all~} s\in S_i;
\\
&
x\succsim^{lex}_i x' && \text{otherwise.} 
\end{align*}
%\begin{itemize}
%\item
%$$
%and
%, or ---
%\item
%$$ and
%($u_s(x)\leq 0$ or $u_s(x')\leq 0$ for some $s\in S_i$).
%\end{itemize}

\begin{theorem}
\label{positive-collective-ee}
If all selves have strictly convex preferences, there exist unambiguous\-/FS and $\succsim^{Nash}_f$\-/aggregate\-/EE Pareto optima.
\end{theorem}

\begin{proof}
Define $\allocation^*$ as the allocation in $E$ that maximizes the minimal Nash bargaining utility:
\begin{align*}
\allocation^*\in \arg\max_{\allocation\in E}\min_{i\in I}U^{Nash}_i(x_i)~~~
\end{align*}
Such an allocation exists since
 each  $U^{Nash}_i$ is continuous on the - compact and non-empty - set $E$. Define $t^*\colon=1+\min_{i\in I} U^{Nash}_i(x^*_i)$. Since $\allocation^*\in E$, $t^*-1\geq 0$.

\bigskip

\textbf{Claim 1:
There is no allocation $\allocation \in E$
for which
$U^{Nash}_i(x_i)\geq t^*-1$ for all $i$ and $U^{Nash}_j(x_j)>t^*-1$ for some $j\in I$.} Suppose $E$ contained such an $\allocation$.  Since $U^{Nash}_j$ is continuous, there exists a small bundle $\epsilon>0$ such that $U^{Nash}_j(x_j-\epsilon)>t^*-1$ and $u_s(x_j-\epsilon)>0$ for each $s\in S_j$.
 Evenly redistribute $\epsilon$ from agent $j$ to all others to obtain the allocation $\allocation^{\circ}$.
%  in the interior of $E$.
 %By construction $U^{Nash}_j(x^{\circ}_j)> t^*-1$.
The strict monotonicity of all selves' preferences yields $U^{Nash}_i(x^{\circ}_i)> t^*-1$  for all $i\neq j$.
In sum we get
   $\min_{i\in I}U^{Nash}_i(x^{\circ}_i)>t^*-1$ ---
    a contradiction to the definition of $\allocation^*$.

\textbf{Claim 2: $\allocation^*$ is unambiguous-FS.} Follows from $\allocation^*\in E$.

\textbf{Claim 3: If $\allocation^*\neq \equalsplit$ then 
$\allocation^*$ is in the interior of $E$, that is,
$\allocation^*$ yields strictly positive utilities to all selves.}
% EREL: Explained what it means to be ``in the interior of E''.
Suppose for contradiction that $\allocation^*\neq \equalsplit$ but $u_s(x^*_i)=0$ for some agent $i$ and  $s\in S_i$.
This implies that
$U^{Nash}_i(\allocation^*_i)=0$,
and $t^*-1=\min_{i\in I}U^{Nash}_i(\allocation^*_i)=0$. Define $\allocation'\colon=\frac{1}{2}\allocation^*+ \frac{1}{2}\equalsplit$.
Since $\allocation^*\neq \equalsplit$, there exists an agent $j$ with $x^*_j\neq \averagebundle$.
By the strict convexity of all selves' preferences, $U^{Nash}_i(x'_i)\geq U^{Nash}_i(\equalsplit)=\min_{i\in I}U^{Nash}_i(x^*_i)=0$ holds for all $i\in I$ and $U^{Nash}_j(x'_j)>0$ --- a contradiction to Claim 1.

\textbf{Claim 4: $\allocation^*$ is Pareto-optimal.}
If some allocation $\allocation$ Pareto-dominates $\allocation^*$, then $U^{Nash}_i(x_i)\geq U^{Nash}_i(x^*_i)\geq \min_{i\in I}U^{Nash}_i(x^*_i)=t^*-1$ for all $i$ and  $U^{Nash}_j(x_j)> U^{Nash}_j(x^*_j)\geq \min_{i\in I}U^{Nash}_i(x^*_i)=t^*-1$ for some $j\in I$, contradicting Claim 1.

\textbf{Claim 5: $\allocation^*$ is $\succsim^{Nash}$-aggregate-EE.} If $\allocation^*=\equalsplit$, then  each agent gets exactly $\averagebundle$, and $\allocation^*$ is even unambiguous-EE.
So assume that $\allocation^*\neq \equalsplit$. By Claim 3,
$\allocation^*$ is in the interior of $E$.
% $\min_{i\in I}U^{Nash}_i(x^*_i)=t^*-1>0$.
 By Claim 1,  $U^{Nash}_i(x^*_i)=t^*-1= U^{Nash}_i(t^*\averagebundle)$ holds for all agents $i\in I$. Since $U^{Nash}_i$ represents the aggregator $\succsim^{Nash}_i$ in the interior of $E$, $x^*$ is $\succsim^{Nash}$\-/aggregate\-/EE with the reference bundle $t^*\averagebundle$.
 \end{proof}

The proof of Theorem \ref{positive-collective-ee} uses two properties of the Nash aggregator:
 Firstly, the continuity of $U^{Nash}_i$ on $E$ is required for  $\allocation^*$ to be well-defined. Secondly,  $U^{Nash}_i(x^*_i)>0$ holds if and only if $u_s(x^*_i) > 0$ for all $s\in S_i$,
 so the allocation $\allocation^*$ is either $(\averagebundle, \dots, \averagebundle)$
or in the interior of $E$ (Claim 3). 

The proof holds unchanged for any aggregator with these two properties.
For example,
$U_i(x_i) := \int_{t_1=0}^{u_1(x_i)} \ldots \int_{t_k=0}^{u_k(x_i)} f_i(t_1,\ldots,t_k)$ for
any positive measurable function $f_i$ of $k$ variables, where $S_i=\{u_1,\dots, u_k\}$.
The function $U^{Nash}_i$ corresponds to the special case $f_i\equiv 1$.%
\footnote{
We are grateful to Magma and Martin R. for these insights.
%https://math.stackexchange.com/q/4305381/29780
}
Other natural aggregators, such as the leximin or the utilitarian aggregator, violate one of these properties. %For example, the leximin aggregator violates property 1
%cannot be represented by continuous functions in the interior of $E$,
%and the utilitarian aggregator violates property 2.
Similarly, if we consider egalitarian equivalence according to one self per agent, then the second property is violated.
We do not know for which alternative aggregators there always exist unambiguous-FS aggregate-EE Pareto-optima.

\subsection{Weakly convex preferences}

Proposition \ref{negative-collective-ee} below shows that
economies where selves have merely convex (not strictly convex) preferences need not have any unambiguous-FS and aggregate-EE Pareto optima. We can however guarantee any two out of these three properties:

\begin{proposition}
\label{negative-collective-ee}
When all selves have  (weakly) convex preferences:
(a)
Some three-agent economies have no unambiguous-FS and aggregate-EE Pareto optimum.
(b) There always exist aggregate\-/EE Pareto optima, unambiguous\-/FS Pareto optima, and aggregate\-/EE unambiguous\-/FS allocations.
\end{proposition}

\begin{figure}
\floatbox[{\capbeside\thisfloatsetup{capbesideposition={left,top},capbesidewidth=.25\textwidth}}]
{figure}
[\FBwidth]
{
\caption[]{
\label{fig:poiee}
\small
Proof of Proposition \ref{negative-collective-ee}.
The solid green line is $y+z=2$; in any unambiguous-FS allocation, all bundles must be on that line.
}
}
{
\includegraphics[width=.7\textwidth]{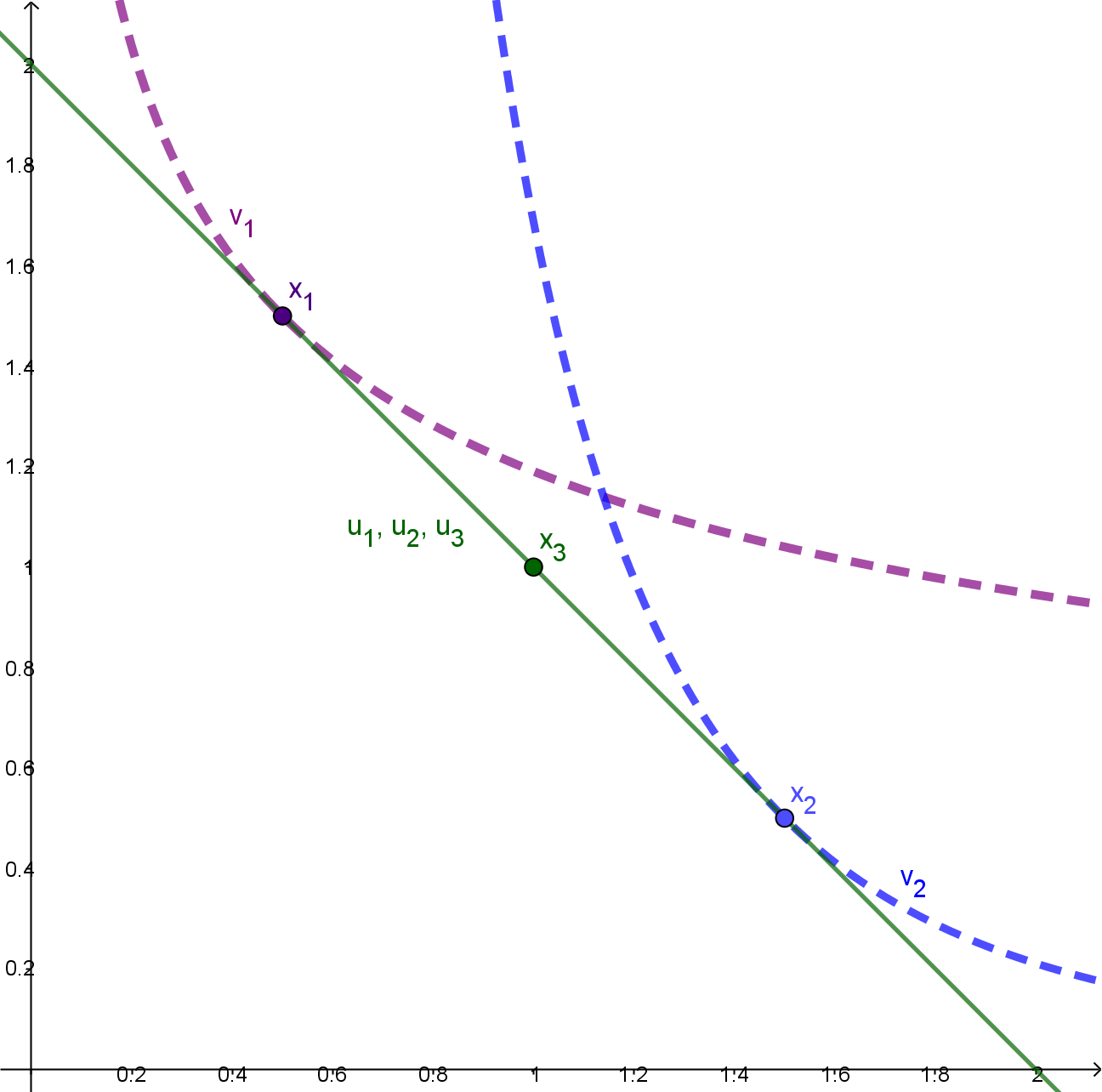}
}
\end{figure}

%\begin{figure}
%\begin{center}
%\includegraphics[width=.9\textwidth]{poiee}
%\end{center}
%\caption{
%\label{fig:poiee}
%Proof of Proposition \ref{negative-collective-ee}.
%The solid green line is $y+z=2$; in any unambiguous-FS allocation, all bundles must be on that line.
%}
%\end{figure}

\begin{proof}
(a)
%Consider a three-agents two-goods economy, where the preferences of the fully rational agent 3 are represented by 
Say 
$u_3(y,z)=y+z$
represents the preferences of the fully rational agent 3.
Agents 1 and 2 have two selves each. The preference of their first selves coincide with agent 3's, and we have  $u_{1}(y,z) = u_{2}(y,z) = u_3(y,z) = y+z$. The two other selves of agents 1 and 2 are represented by  different Cobb-Douglas utilities: $v_1(y,z) = y^{1/4}z^{3/4}$ and $v_2(y,z)=y^{3/4}z^{1/4}$.
For the fair-share guarantee to hold according to the three selves with identical linear preferences, $y_1+z_1=y_2+z_2=y_3+z_3=2$ must hold.
The only Pareto-optimal allocation satisfying these equations gives agent 1 $(1/2,3/2)$ and agent 2 $(3/2,1/2)$; see Figure \ref{fig:poiee}.

Now suppose the allocation was aggregate-EE for some aggregator and
reference bundle $r:=(y_r,z_r)$.
For agent 3  to be indifferent between his bundle $(1,1)$ and $r$, $y_r+z_r$ must equal 2, so $r$ must lie on the solid green line in Figure \ref{fig:poiee}.
This implies that the two other agents are according to their linear selves indifferent between their  bundles and $r$. But, for any $r$ with $y_r+z_r=2$, at least one agent prefers her  bundle to $r$ according to her Cobb-Douglas self.
This agent must according any aggregator strictly prefer her bundle to $r$.%
\footnote{
To see where the proof of Theorem \ref{positive-collective-ee} fails on this  economy,
%normalize all utilities to assign value 0 to the bundle $(1,1)$. So $U^{Nash}_1(y,z) = (y+z-2)\cdot(y^{1/4}z^{3/4}-1)$ and $U^{Nash}_2(y,z) = (y+z-2)\cdot(y^{3/4}z^{1/4}-1)$ and $U^{Nash}_3(y,z) = (y+z-2)$.
note that $E$ contains only allocations $\allocation$ with $y_i+z_i=2$ for all agents $i$ (the set represented by the solid green line). So  $U^{Nash}_i(x_i)=0$ holds for all $i \in I$ and all $\allocation \in E$, and any $\allocation\in E$ satisfies the definition of $\allocation^*$. In particular, Claim 3 fails.
}
\bigskip

(b)
Theorem \ref{positive-pofs} proved the existence of unambiguous\-/FS Pareto optima. The equal allocation is aggregate\-/EE and unambiguous\-/FS.
To find an aggregate-EE Pareto optimum without unambiguous-FS, replace the Nash aggregator in the proof of Theorem \ref{positive-collective-ee}
with the utilitarian aggregator.
As this aggregator is represented by a continuous function on the entire set $X$ (not only in the interior of $E$),
the proof of aggregate-EE is now valid without the need for Claim 3.
\end{proof}

\section{Intransitive aggregators}\label{sec:intrans}
An alternative approach towards fairness in the context of behavioral economics would explicitly consider some intransitive aggregations of the agents' selves.
Agents might aggregate their objectives intransitively, using voting rules
  or by sequential procedures,  as suggested by respectively  \citet{green2007choice}  and \citet{ApesteguiaBallester}. 
For a concrete example, say an allocation is \emph{majority-fair} if at least half of the selves of any given agent deem it to be fair.\footnote{
Majority fairness was studied in the context of cake-cutting \citep{SegalHalevi2018Families}
and indivisible item allocation
\citep{segal2019democratic}.
} So an allocation is majority envy free (majority-EF) if no more than half of any agent's selves envy a different agent, it is majority egalitarian equivalent (majority-EE) if there exists a references bundle $r$ such that at least half of the selves of any given agent are indifferent between $r$ and their bundle.

Our negative results can be extended to show that majority-fairness is incompatible with efficiency.
For a parallel to the our result on unambiguous-EF results Theorem \ref{negative-pone}, consider an economy with four rational agents ($1,2,3,4$) and one boundedly rational agent 5 with three selves ($5a,5b,5c$).
Say that the utilities of all selves have the single-crossing property. The indifference-curves of the rational agents 1 and 2 are between the indifference curves of the  selves $5a$ and $5b$; the indifference curves of the remaining two rational agents 3 and 4 are between the indifference curves of selves $5b$ and $5c$. Suppose a majority-EF efficient allocation $\allocation^*$ exists and that all agents consume positive quantities of each good in this allocation. Due to the single crossing property, no two rational agents can consume the same quantity in the allocation $\allocation^*$.
In any majority-EF allocation, at most one self of agent 5 may envy any other agent.
So say without loss of generality that self $5a$ as well as self $s\in\{5b,5c\}$ do not envy any other agent.
Now we have a situation similar to the example used to prove Proposition \ref{negative-pone}:
No envy, together with the fact that agent 1's indifference curves lie between the indifference curves of selves $5a$ and $s$, imply that agents 1 and 5 must consume the same bundle. Mutatis mutandis, we see that also agent 2 and agent 5 have to consume the same bundle --- a contradiction to agents 1 and 2 consuming different bundles in any interior Pareto optimum.

Clearly, the dearth of unambiguous-EF allocations is owed to the fact that the unambiguous ranking  according to which each agent must prefer their own bundle to any other agent's might be highly incomplete.  Our notion of majority-EF then presents an interesting contrast: even though majority preference is a complete ranking, the example used to prove the non-existence result for unambiguous-EF transfers with some modification to the case of majority-EF allocations.  So we see that both - completeness and transitivity - play their role in existence of envy free Pareto optima in rational environments.

Similarly, our example in the proof of result that unambiguous-EE need not exist (Theorem \ref{negative-poee}) can be extended to majority egalitarian equivalence:
 Consider an economy with 2 agents, each of whom has 3 selves. Say all utilities in the model satisfy the single-crossing property.
In any majority-EE allocation, at least two selves  of each agent must be indifferent between their bundle and the reference bundle $r$, so the situation is similar to the one analyzed in Proposition \ref{negative-poee}.

\section{Related work}

\subsection{Behavioral Welfare Economics}
While there is a large literature on boundedly rational choices by single agents,
our paper belongs to the much smaller literature on
welfare economies with two or more boundedly rational agents.
\citet{BernheimRangel2009} define several notions of Pareto-optimality and of competitive equilibrium
for economies with multi-self agents,
and prove that  competitive equilibria always satisfy one of these notions of Pareto-optimality.
% Their theorem is a special case of the welfare theorem of \citet{FonOtani1979} for an economy of agents with incomplete and intransitive preferences.
\citet{mandler2014indecisiveness}
proves that, with multi-self agents, the set of Pareto-optima might be very large, in the sense that it has the same dimension as the set of all allocations. Since almost all Pareto-optima lie in full-dimensional neighborhoods of other Pareto optima, he argues  that the Pareto criterion is of limited use for policy decisions. To close this gap of indecisiveness, \citet{mandler2020distributive} suggests to use utilitarian-optimality.
\citet{danan2015harsanyi} characterize a utilitarian social-choice rule when agents have incomplete preferences,
and \citet{FleurbaySchokkaert2013} characterizes an egalitarian social-choice rule --- based on the principle of income-equivalence --- with incomplete preferences.

\subsection{Fairness criteria}
The modern study of fairness in economics was initiated by
\citet{Steinhaus1948}, who
proved the existence of \emph{fair-share} allocations of a heterogeneous good (``cake''). Since then, fairness has been extensively studied in economics \citep{young1995equity,Moulin2004Fair,Thomson2011Fair} as well as in other disciplines such as mathematics and computer science.

The existence of \emph{envy-free} Pareto-optima was initially proved as a consequence of the existence of market equilibrium \citep{foley1967resource}.
%If all preferences are convex and strictly monotonic, then there exists a market equilibrium from equal incomes, and it is both Pareto-optimal and envy-free.
\citet{varian1974equity} showed that no-envy Pareto-optima may exist even when a competitive equilibrium does not.
\citet{varian1974equity} replaces assumption of convex preferences with a condition on the set of allocations associated with any
 weakly-Pareto-optimal vector of utilities. If any such set is a singleton an envy free Pareto optimum exists.
\citet{svensson1983existence} and  \citet{Diamantaras1992Equity} then respectively
showed that \citet{varian1974equity}'s ``singleton''-requirement can be replaced by convexity and contractibility.
 \citet{Diamantaras1992Equity}'s result  applies to economies with public goods.
\citet{svensson1994sigma} and
 \citet{Bogomolnaia2017Competitive} showed that envy free Pareto optima exist under yet further relaxed conditions.

The \emph{egalitarian equivalence} criterion
was introduced by \citet{pazner1978egalitarian}. They argue that the equal split $\equalsplit$ where each agent gets the same bundle is, from an egalitarian perspective, an ideal division. Since $\equalsplit$
 is usually not efficient, \citet{pazner1978egalitarian} propose to regain efficiency by considering all allocations for which there exists an (infeasible) allocation where all agents consume the same bundle, such that each agent is indifferent between that bundle and their consumption in the allocation.
\citet{pazner1978egalitarian} proved that  egalitarian-equivalent Pareto optima exist
 in economies with production, which may not have any envy-free Pareto-optima \citep{vohra1992equity}.
\citet{thomson1990non} showed that egalitarian equivalent Pareto optima  are generically  not envy-free when there are at least three agents.

\subsection{Fair division among families}
Several papers study fair division with multi-self agents under the interpretation of agents and their selves as families and their members. 
\citet{SegalHalevi2018Families}
study fair division in cake-cutting problems, where the challenge is to allocate connected pieces, or at least pieces made of a small number of connected components.
\citet{ManurangsiSu17,Suksompong18,segal2019democratic,kyropoulou2019almost} study fair division in problems with indivisible goods.  Since there may be no fair allocation of indivisible goods, the focus in these studies is on finding approximately-fair allocations.
In contrast, our model of fully-divisible goods always allows for fair allocations, and the challenge is to find allocations that are both fair and efficient.
\citet{ghodsi2018rent} studies fair division of rooms and rent among families of tenants, using three notions of fairness termed strong, aggregate and weak. Their strong-fairness is our unambiguous no-envy; their weak-fairness holds if at least one self per agent does not envy; their aggregate-fairness  corresponds to our aggregate fairness with the utilitarian aggregator.

%It is important to distinguish our fairness notions from two very different notions of \emph{group fairness}.

%(a) One notion of group-fairness involves
%the standard resource-allocation setting in which each agent receives an individual bundle
%\citep{Berliant1992Fair,Husseinov2011Theory,DallAglio2009Cooperation,DallAglio2014Finding,todo2011generalizing,mouri2012envy}.
%A \emph{group-envy-free} division is defined as a division in which no coalition of agents can take the pieces allocated to another coalition with the same number of agents and re-divide the pieces among its members such that all members are weakly better-off.
%In our setting the families are fixed and agents do not form coalitions on-the-fly; the challenge arises from the fact that all individuals in each family consume the same bundle.
%We use the term ``families'' to emphasize that they are pre-determined and agents do not form coalitions on-the-fly.

%(b) A second notion of group-fairness comes from an entirely different field --- artificial intelligence.
%Consider an AI system that automatically detects potential criminals based on their personal traits.
%If such a system reports significantly more suspects with a certain skin-color, this might be considered a violation of group-fairness --- the members of the group with that %particular skin-color are treated unfairly. There is a growing literature on various definitions of group-fairness in this context;
%see, for example, \citet{Dwork2012FTA}, \citet{Herbert2017Calibration} and the references therein.

\subsection{Public and club goods}
The existence of fair allocations with public and private goods has been studied e.g. by \citet{Diamantaras1992Equity,Diamantaras1994Generalization,Diamantaras1996Set,Guth2002NonDiscriminatory}.
With multi-self agents, our fairness notions correspond to \emph{club goods} \citep{buchanan1965economic}, that is, goods that
are public among all selves of the agent --- but  private outside. As far as we know, fair allocation of goods among different clubs has not yet been considered.
Under the interpretation of multi-self agents as clubs,
 questions such as the optimal number of members in a club, optimal number of clubs, optimal quantity of club-good provision, pricing policies and exclusion mechanisms have been studied \citep{sandler1980economic,hillman1993socialist,sandler1997club,loertscher2017club,mackenzie2018club}.
%Some papers have studied fairness in the context of discrimination prevention, i.e., when competition between different clubs prevents clubs from discriminating against some of its members by charging them higher fees \citep{sandler1980economic}.
A modern example of a club good is information: information is partly excludable (via intellectual property law), but once it is given to a group, it is not rival. Therefore our work may have implications on fair division of information, for example, dividing training samples for machine learning among groups of researchers.

\section{Acknowledgments}
We are grateful for the support of the Minerva foundation through the ARCHES prize.
Erel is grateful to the Israel Science Foundation for grant 712/20.

The paper started as a discussion in the economics stackexchange website.%
\footnote{https://economics.stackexchange.com/q/9916/385}
We are grateful to Shane Auerbach, Amit Goyal and Kitsune Cavalry for participating in the discussion.

We are grateful to Hal Varian, William Thomson, Herve Moulin, Gerald A Edgar, Martin R., and reviewers of Games and Economic Behavior for their very helpful feedback on a previous version of this paper.

\newpage

\bibliographystyle{apalike}
\bibliography{merged}
\end{document}